\documentclass[12pt,a4paper]{article}
\usepackage{amsmath,amsfonts,amssymb,amsthm, epsfig,
epstopdf, titling, url, array}
\usepackage[utf8]{inputenc}
\usepackage{mathtools}
\usepackage[english]{babel}
\usepackage{hyperref}
\newtheorem{thm}{Theorem}[section]
\newtheorem{lem}{Lemma}[section]
\newtheorem{cor}{Corollary}[section]
\newtheorem{prop}{Proposition}[section]

\numberwithin{equation}{section}

\def\R{I\!\!R}

\title{
\textbf{Finding formulas connecting Bessel and hypergeometric functions using multipliers of Bessel operator I} }

\author{Mohamed Vall Ould Moustapha}
\begin{document}

\maketitle

\begin{abstract}In this work we give explicit formulas for the Schwartz integral kernels of some multipliers of the Bessel operator on $\R^\ast_+$. By using the integral transforms connecting these multipliers we obtain old and new formulas involving Bessel and hypergeometric functions.
\end{abstract}
\begin{center}\bf
\section{Introduction}
\end{center}
The Bessel operator is an interesting operator which arises in several contexts,  one of them being the Schr\"odinger equation in non relativistic quantum mechanics \cite{15}. For the recent papers on the Bessel operator see(\cite{3, 9, 10, 17}).
The aim of this paper is twofold: first we give explicit formulas for the Schwartz integral kernels of the following multipliers
$ H_\nu^p(t):=e^{t L_\nu}(\sqrt{-L_\nu})^p, R_\nu^p(\lambda):=\left(L_\nu+\lambda^2\right)^{-1}(\sqrt{-L_\nu})^p, R_\nu^{\mu, p}(\lambda):=\left(L_\nu+\lambda^2\right)^{-1-\mu}(\sqrt{-L_\nu})^p$ called here respectively the weighted heat, weighted resolvent and weighted generalised resolvent operator associated to the  Bessel operator on $\R^\ast_+$:\\
 
\begin{equation}\label{(1.5)} L_{\nu}=\frac{\partial^2}{\partial x^2} +\frac{1/4-\nu^2}{x^2}, \ \ \ \ \ \ \nu\in \R .\end{equation}
and secondly, by using the integral transforms connecting these multipliers we obtain old and new formulas involving Bessel and hypergeometric functions. For finding formulas involving Bessel and hypergeometric functions using multipliers on the Euclidian space $\R^n$ see\cite{13}. \\
First of all we recall the following formulas for the classical heat kernel for the Bessel operator (\cite{1}, p. 68):
\begin{equation}\label{cl-heat}H^0_\nu(t, x, x')=\frac{(x x')^{1/2}}{2 t}e^{\frac{-(x^2+x'^2)}{4t}}I_\nu(\frac{x x'}{2t}).\end{equation}
 and the resolvent kernel \cite{10}
\begin{equation}\label{cl-res} R_\nu^0(\lambda, x, x')=\frac{i\pi}{2}\sqrt{x x'}\left\{\begin{array}{cc}J_\nu(\lambda x) H^{(1)}_\nu(\lambda x')& \mbox{$x<x'$}\\
J_\nu(\lambda x') H^{(1)}_\nu(\lambda x)&\mbox{$x>x'$}
\end{array}\right.\end{equation}
Using the fact that the resolvent kernel is the Laplace transform of the heat kernel we have for $ {\cal R} e \lambda^2 < 0$ and $\nu>-1$:

\begin{equation}\label{lapl-heat}\int_0^\infty e^{\lambda^2 t} t^{-1}e^{\frac{-(x^2+x'^2)}{4t}}I_\nu(\frac{x x'}{2t}) dt=i\pi \left\{\begin{array}{cc}J_\nu(\lambda x) H^{(1)}_\nu(\lambda x')& \mbox{$x<x'$}\\
J_\nu(\lambda x') H^{(1)}_\nu(\lambda x)&\mbox{$x>x'$}
\end{array}\right.\end{equation}
where $I_\nu$ is the first kind modified Bessel function, $J_\nu$ and $H^{( 1 )}_\nu$ are respectively the first and the third kind Bessel functions ( see\cite{5, 11} ).\\
We mention that the absolute convergence of the above integral is assured by the formulas  (\cite{11}, p. 136)
\begin{equation}I_\nu(x)\approx \frac{x^\nu}{2^\nu\Gamma(1+\nu)}, x\rightarrow 0; 
I_\nu(x)\approx \frac{e^x}{\sqrt{2 \pi x}}, x\rightarrow \infty .\end{equation} 
The end of this section is devoted to the preliminaries on the Hankel transform on $\R^+$.\\

For $\nu>-1$, the Hankel transform of order $\nu$ for a function $f\in C^\infty_0(\R^+)$,is defined by the integral
\begin{equation}\label{4} \left(H_\nu f\right)(\omega)=\int_0^\infty (x\omega)^{1/2}J_\nu(x\omega)f(x)d x\end{equation}
where $J_\nu$ is the first order Bessel function of order $\nu$.\\

\begin{prop}\cite{14} For $\nu>-1$, we have\\
$i) H_\nu^2=1$\\
$ii) H_\nu $ is self adjoint\\
$iii) H_\nu$ is an $L^2$  isometry \\
$iv) H_\nu L_\nu=-\omega^2H_\nu$.\\
\end{prop}
For more informations on the Hankel transform the reader can consults the nice book by Davies \cite{2}. \\
Note that we can define  $\phi (\sqrt{-L_\nu})$ for $\phi$ a well behaved Borel function by using the  Hankel transform:\\
\begin{prop} For $\nu> -1$, the Schwartz integral kernel of the operator $\phi(\sqrt{-L_\nu})$, is given at last formally by

\begin{equation}\label{spectral}K_{\nu}(\phi, x, x')=(x x')^{1/2}\int_0^\infty J_{\nu}(\omega x)J_\nu(\omega x')\phi(\omega) \omega d\omega. \end{equation}
\end{prop}

 The proof of this proposition uses essentially Proposition 1.1  and in consequence is left to the reader.\\ 
Note that using \eqref{spectral} with $\phi(x)= \frac{1}{-\omega^2+\lambda^2}$ we obtain 
for ${\cal R}e \lambda^2<0$ and $\nu>-1$  the following formula
\begin{equation}\label{5}\int_0^\infty \frac{J_\nu(x\omega) J_\nu(x' \omega)}{-\omega^2+\lambda^2}\omega d\omega=\frac{i\pi}{2}\left\{\begin{array}{cc}J_\nu(\lambda x) H^{(1)}_\nu(\lambda x')& x<x'\\
J_\nu(\lambda x') H^{(1)}_\nu(\lambda x)& x>x'
\end{array}\right.\end{equation}
The absolute convergence of the above integral is assured by the formulas
 (\cite{11}, p.134)
\begin{equation}\label{as2}J_\nu(x)\approx \frac{x^\nu}{2^\nu\Gamma(1+\nu)}, x\rightarrow 0,
J_\nu(x)\approx \sqrt{\frac{2}{\pi x}},\, x\rightarrow \infty.\end{equation}
The following lemma gives the Laplace transform of the two variables Humbert confluent hypergeometric  function 
 (\cite{5}, p.225):
\begin{equation}\label{humbert}\Psi_2\left(a;c,c',x,y\right)=\sum_{n;m\geq 0}\frac{(a)_{n+m}}{(c)_m (c')_n m! n!} x^m y^n, \hfil |x|<\infty,  |y|<\infty.  \end{equation}
in term of the Kamp\'e de Feri\'et generalized hypergeometric function
 $F^{A: B}_{C: D}$ given by EXTON (\cite{7}, p.29). 
 \begin{align}\label{kampe}F^{A: B}_{C: D}\left({}^{a_1,...a_A: b_1,...b_B; b'_1,...b'_B}_{c_1,...,c_C: d_1,...d_D; d'_1,...d'_D}; x, x'\right)=\nonumber\\
\sum_{n, m\geq 0}\frac{\prod_{j=1}^{A}(a_j)_{m+n}\prod_{j=1}^{B}(b_j)_{m}\prod_{j=1}^{B}(b'_j)_{n}}{\prod_{j=1}^{C}(c_j)_{m+n}\prod_{j=1}^{D}(d_j)_{m}\prod_{j=1}^{D}(d'_j)_{n}m! n!} x^m x'^n. \end{align}
with $A+B <C+D+1$, $|x|<\infty$ and $|x'|<\infty$.\\
\begin{lem}\label{key-lemma} For $\gamma>0$, $\alpha>0$, $X, Y\in \R$ we have
\begin{align} \label{lap-humbert}
\int_0^\infty e^{-\gamma t}t^{\alpha-1}\Psi_2\left(a, b_1, b_2, \frac{X}{t}, \frac{Y}{t}\right) dt=\nonumber\\
\frac{\Gamma(\alpha)}{\gamma^\alpha}F^{1:  0}_{1: 1}\left({}^{a: -, -}_{1-\alpha, b_1, b_2}; -\gamma X, -\gamma Y\right).
 \end{align}
\end{lem}
\begin{proof}
Replacing the confluent hypergeometric $\Psi_2$ by its series  \eqref{humbert} in the integral  \eqref{lap-humbert} and integrating term by term we obtain\\
 \begin{align}\int_0^\infty e^{-\gamma t}t^{\alpha-1}\Psi_2\left(a, b_1, b_2, \frac{X}{t}, \frac{Y}{t}\right) dt= \nonumber
\gamma^{-\alpha}\sum_{n, m\geq 0}\frac{(a)_{n+m}\Gamma(\alpha-n-m)}{(b_1)_m (b_2)_n m! n!} (\gamma x)^m (\gamma y)^n.
 \end{align}
using the formula (\cite{16}, p. 22)  $\frac{\Gamma(\alpha-n)}{\Gamma(\alpha)}=\frac{(-1)^n}{(1-\alpha)_n}$
we can write\\
 \begin{align}\int_0^\infty e^{-\gamma t}t^{\alpha-1}\Psi_2\left(a, b_1, b_2, \frac{X}{t}, \frac{Y}{t}\right) dt=\nonumber\\
\frac{\Gamma(\alpha)}{\gamma^\alpha}\sum_{n, m\geq 0}\frac{(a)_{n+m}}{(1-\alpha)_{n+m}(b_1)_m (b_2)_n m! n!} (-\gamma x)^m (-\gamma y)^n.
 \end{align}
which gives the result in \eqref{lap-humbert} and the proof of  Lemma \ref{key-lemma} is finished.
 \end{proof} 
The organization of the remaining of the paper is as follows,
 the Schwartz integral kernel of the weighted heat evolution operator
$e^{t L_\nu}(\sqrt{-L_\nu})^p$ and the weighted Schr\"odinger evolution operator
$e^{i t L_\nu}(\sqrt{-L_\nu})^p$ will be given in section 2. In section 3 we will obtain a closed form of the Schwartz integral kernel of the weighted resolvent operator. The section 4  is devoted to the Schwartz integral kernel of the weighted generalized resolvent  operator on $\R^+$. 

\section{Weighted Heat evolution operator for Bessel operator on $\R^+$}

In this section we give the Schwartz  integral kernels of the weighted heat and Schr\"odinger evolution operators $e^{t L_\nu}(\sqrt{-L_\nu})^p$ and $e^{i t L_\nu}(\sqrt{-L_\nu})^p$ in explicit forms.
\begin{thm} For ${\cal R} e p>-2(\nu+1)$ and $\nu > -1$, the Schwartz integral kernel $H_\nu^p(t, x, x')$ of the weighted heat evolution operator $e^{t L_\nu}(\sqrt{-L_\nu})^p$  is given  as:\\
\begin{align}H_{\nu}^p(t, x, x')=\frac{\Gamma(p/2+1+\nu)(x/2)^{\nu+1/2}(x'/2)^{\nu+1/2}}{[\Gamma(\nu+1)]^2t^{p/2+\nu+1}}\times\nonumber\\ \label{w-heat}\Psi_2\left(p/2+1+\nu,\nu+1,\nu+1; x^2/4 t, x'^2/4 t\right) \end{align}
The function $\Psi_2\left(a,c,c'; x;y\right)$ denotes the Humbert's confluent hypergeometric function of two variables given in  \eqref{humbert}.
\end{thm}
\begin{proof} Using the formula \eqref{spectral} with $\phi(\omega)=e^{-t\omega^2}\omega^p$  we have
\begin{equation} H_\nu^p(t, x, x')=(x x')^{1/2}\int_0^\infty J_{\nu}(\omega x)J_\nu(\omega x')e^{-t\omega^2}\omega^{p+1} d\omega.\end{equation}
 Next we employ the formula, $M=\mu_1+\mu_2$, $ {\cal R}e (\nu+M)>0$,( \cite{6} p.187 )
\begin{align}\label{lap-prod-bess} \nonumber \int_0^\infty e^{-p t} t^{\nu-1} J_{2\mu_1}(2(a_1 t)^{1/2}) J_{2\mu_2}(2(a_2 t)^{1/2})dt=\\ \frac{\Gamma(\nu+M)}{\Gamma(2\mu_1+1)\Gamma(2\mu_2+1)} p^{-\nu-M}a_1^{\mu_1} a_2^{\mu_2}
\Psi_2\left(\nu+M,2\mu_1+1,2\mu_2+1, a_1/p, a_2/p\right). \end{align}
 and we arrive at the formula  \eqref{w-heat}.
\end{proof}
\begin{cor} The Schwartz integral kernel $K_\nu^p(t, x, x')$ of the weighted Schr\"odinger evolution operator with inverse square potential $e^{i t L_\nu}(\sqrt{-L_\nu})^p$  is given in terms of the two variables Humbert's confluent hypergeometric function  for ${\cal R} e p>-2(\nu+1)$  and $\nu > -1$ as\\
\begin{align}K_{\nu}^p(t, x, x')=\frac{\Gamma(p/2+1+\nu)(x/2)^{\nu+1/2}(x'/2)^{\nu+1/2}}{[\Gamma(\nu+1)]^2(it)^{p/2+\nu+1}}times\nonumber\\ \Psi_2\left(p/2+1+\nu,\nu+1,\nu+1; x^2/4it, x'^2/4it\right). \end{align}
\end{cor}
By taking $p=0$ in Theorem 2.1 we have\\
\begin{align}H_{\nu}^0(t, x, x')=\frac{(x/2)^{\nu+1/2}(x'/2)^{\nu+1/2}}{\Gamma(\nu+1)t^{\nu+1}}\times\nonumber\\    \Psi_2\left(\nu+1,\nu+1,\nu+1; x^2/4 t, x'^2/4 t\right).\end{align}
and by comparing this with  \eqref{cl-res} we have
\begin{equation}\Psi_2\left(\nu+1,\nu+1,\nu+1; x, y\right) =\Gamma(\nu+1)(x y)^{-\nu/2}e^{-x-y}I_\nu(2\sqrt{ x y}).\end{equation}

\section{Weighted resolvent operator for the Bessel operator on $\R^+$}
In this section we give explicit formula for the Schwartz integral kernel of the weighted resolvent operator
$R_\nu^p(\lambda)=\left(L_\nu+\lambda^2\right)^{-1}(-L_\nu)^{p/2}$
using the formula
\begin{equation}\label{lap-heat}R_\nu^ p(\lambda, x, x')=\int_0^\infty e^{\lambda^2 t}H_\nu^p(t, x, x') dt, \ \ \ \ \ \ \ {\cal R} e\lambda^2<0.\end{equation}
where $ H_\nu^p(t, x, x'$) is the Schwartz integral kernels of
the weighted heat operator.\\
 The Formula \eqref{lap-heat} is a consequence of the formula $(a^2+y^2)^{-1}=\int_0^\infty e^{- (a^2+y^2)t}\, dt$  valid for  ${\cal R}e a^2>0$.

\begin{thm}For ${\cal R} e \lambda^2 < 0$, $-1 < p/2+\nu < 0$  and $\nu > -1$, \\
the Schwartz integral kernel for the weighted resolvent operator$\left(L_\nu+\lambda^2\right)^{-1}(-L_\nu)^{p/2}$ is given by\\
$ R_\nu^p(\lambda, x, x')=\frac{[\Gamma(p/2+\nu+1)\Gamma(-p/2-\nu)] }{[\Gamma(\nu+1)]^2}\times$ 
\begin{equation}\label{wr} (-\lambda^2)^{p/2+\nu}(x x'/4)^{\nu+1/2}
F^{0:0}_{0:1}\left({}_{-: \nu+1;\nu+1}^{-: - ; -};\frac{\lambda^2 x^2}{4},\frac{\lambda^2 x'^2}{4}\right).\end{equation}
where $F^{0:0}_{0:1}$ is the Kampe de Firiet hypergeometric function given in \eqref{kampe}.
\end{thm}
The proof of this  theorem can be seen as is a direct application of Proposition 3.1, Theorem 2.1 and of Lemma 1.1.
\begin{cor} For ${\cal R} e \lambda^2 < 0$, $-1 < p/2+\nu < 0$  and $\nu > -1$\\
$\int_0^\infty \frac{J_\nu(x\omega) J_\nu(x' \omega)}{-\omega^2+\lambda^2}\omega^{p+1} d\omega =
\frac{\Gamma(p/2+\nu+1)\Gamma(-p/2-\nu) }{2[\Gamma(\nu+1)]^2}\times$
 \begin{equation}  (-\lambda^2)^{p/2+\nu}(x x'/4)^{\nu}
F^{0:0}_{0:1}\left({}_{-: \nu+1;\nu+1}^{-: - ; -};\frac{\lambda^2 x^2}{4},\frac{\lambda^2 x'^2}{4}\right)\end{equation}
\end{cor}
\begin{proof}
Using Proposition 1.2 with $\phi(\omega)=(-\omega^2+\lambda^2)^{-1}\omega^p$ and Theorem 3.1,we obtain the result,
where the absolute convergence of the above integral is assured by the formulas \eqref{as2}.
\end{proof}
Note that by taking $p=0$  in \eqref{wr} and comparing with \eqref{cl-res}  the following formula is valid for $x<x'$\\
$J_\nu(\lambda x) H^{(1)}_\nu(\lambda x')= \frac{\Gamma(-\nu) (-\lambda^2  x x'/4)^{-\nu}}{ i\pi \Gamma(\nu+1)}\times$ \begin{equation}F^{0:0}_{0:1}\left({}_{-: \nu+1;\nu+1}^{-: - ; -};\frac{\lambda^2 x^2}{4},\frac{\lambda^2 x'^2}{4}\right).\end{equation}

\section{ Weighted generalized resolvent operator for the Bessel operator on $\R^+$}
In this section we generalize some results of the section $3$  by giving an explicit expression of the  weighted generalized resolvent kernels
$R_\nu^{\mu, p}(\lambda)=\left(L_\nu+\lambda^2\right)^{-1-\mu}(-L_\nu)^{p/2}$.\\
\begin{prop} We have  the following formula connecting the weighted generalized resolvent kernel to the weighted heat kernel
 \begin{equation}\label{w-gr}  R_\nu^{\mu, p}(\lambda, x, x')=\frac{1}{\Gamma(\mu+1)}\int_0^\infty e^{\lambda^2 t}t^\mu H_\nu^p(t, x, x') dt; \ \ \ \ \ \ \ {\cal R} e\lambda^2<0.\end{equation}
\end{prop}
\begin{proof}We use the formula $(a^2+y^2)^{-1-\mu}=\frac{1}{\Gamma(\mu+1)}\int_0^\infty e^{- (a^2+y^2)t} t^\mu\, dt$ for  ${\cal R}e a^2>0$
\end{proof}
\begin{thm} For ${\cal R} e \lambda^2 < 0$, $-1 < p/2+\nu < \mu$, $\nu > -1$  and $\mu > -1$, the Schwartz integral kernel of the weighted generalized resolvent kernel with inverse square potential is given by\\

\begin{align}\label{w-gr} \nonumber R_\nu^{\mu, p}(\lambda,x, x')=\frac{\Gamma(\mu-p/2-\nu)[\Gamma(p/2+\nu+1)]}{\Gamma(\mu+1)[\Gamma(\nu+1)]^2}\times\\ (-\lambda^2)^{p/2+\nu-\mu}(x x'/4)^{\nu+1/2}
  F^{1:0}_{1:1}\left({}_{p/2+\nu+1-\mu: \nu+1;\nu+1}^{p/2+\nu+1: - ; -};\frac{\lambda^2 x^2}{4},\frac{\lambda^2 x'^2}{4}\right)\end{align}
where $F^{A: B}_{C: D}$ is the Kampe de Feri\'et generalized hypergeometric function given by\eqref{kampe}.
\end{thm}
\begin{proof}
This theorem is a direct consequence of Proposition 4.1, Theorem 2.1 and  Lemma 1.1.
\end{proof}
\begin{cor}  For ${\cal R} e \lambda^2 < 0$, $-1 < p/2+\nu < \mu$, $\nu > -1$  and $\mu > -1$,  we have the following formula\\
\begin{align}\int_0^\infty \frac{J_\nu(x\omega) J_\nu(x' \omega)}{(-\omega^2+\lambda^2)^{1+\mu}}\omega^{p+1} d\omega=
\frac{\Gamma(\mu-p/2-\nu)[\Gamma(p/2+\nu+1)]}{\Gamma(\mu+1)[\Gamma(\nu+1)]^2}\times\nonumber \\ (-\lambda^2)^{p/2+\nu-\mu}(x x'/4)^{\nu+1/2}
  F^{1:0}_{1:1}\left({}_{p/2+\nu+1-\mu: \nu+1;\nu+1}^{p/2+\nu+1: - ; -};\frac{\lambda^2 x^2}{4},\frac{\lambda^2 x'^2}{4}\right)\end{align}
\end{cor}
\begin{proof} Using Proposition 1.2  with $\phi(\omega)=(-\omega^2+\lambda^2)^{-1-\mu}\omega^p$ and Theorem 4.1. Note that
the absolute convergence of the integral is assured by the formulas \eqref{as2}
\end{proof}
 By taking $p=0$ in the formula \eqref{w-gr}, we see that the Schwartz integral kernel of  generalized resolvent with inverse square potential is given by\\
\begin{align}R_\nu^{\mu, 0}(\lambda,x, x'))=\frac{\Gamma(\mu-\nu)}{\Gamma(\mu+1)[\Gamma(\nu+1)]}(-\lambda^2)^{\nu-\mu}(x x'/4)^{\nu+1/2}\times\nonumber\\ 
  F^{1:0}_{1:1}\left({}_{\nu+1-\mu: \nu+1;\nu+1}^{\nu+1: - ; -};\frac{\lambda^2 x^2}{4},\frac{\lambda^2 x'^2}{4}\right)\end{align}
where ${\cal R} e \lambda^2 < 0$, $\nu > -1$  and $\mu > -1$ and $\mu>\nu$.\\
By taking  $p+1=\alpha-1$ and $\mu+1=\rho$ we have\\
$\int_0^\infty x^{\alpha-1} \frac{J_\nu(c x) J_\nu(c x)}{(x^2+ z^2)^{\rho}}  d x =(-1)^\rho
\frac{\Gamma(\rho-\alpha/2-\nu)\Gamma(\alpha/2+\nu) }{\Gamma(\rho)[\Gamma(\nu+1)]^2}$
 \begin{equation}  (z^2)^{\alpha/2+\nu-\rho}(c^2/4)^{\nu+1/2}
F^{1:0}_{1:1}\left({}_{\alpha/2+\nu-\rho+1, \nu+1;\nu+1}^{\alpha/2+\nu; - ; -};-\frac{c^2  z^2}{4},-\frac{c^2 z^2}{4}\right)\end{equation}
Using the formulas \cite{8} p. 672-673  we obtain\\
 $$ (-1)^\rho (z^2)^{\alpha/2+\nu-\rho}(c^2/4)^{\nu+1/2}
F^{1:0}_{1:1}\left({}_{\alpha/2+\nu-\rho+1, \nu+1;\nu+1}^{\alpha/2+\nu; - ; -};-\frac{c^2  z^2}{4},-\frac{c^2 z^2}{4}\right)$$
$=\frac{1}{2}(\frac{c}{2})^{2\rho-\alpha}\Gamma \left[{}^{\nu+\alpha/2-\rho, 1+2\rho-\alpha, \rho, \nu+1, \nu+1 }_
{-\alpha/2+\rho+1,\nu-\alpha/2+\rho+1, -\alpha/2+\rho+1, \rho-\alpha/2-\nu, \alpha/2+\nu}\right]\times$\\ ${}_2 F_3((1-\alpha)/2+\rho, 
\rho;  \rho+1-\nu-\alpha/2,\rho+1-\alpha/2;  c^2 z^2)+ \\
\frac{z^{\alpha-2\rho}}{2}(\frac{c z}{2})^{2\nu}{}_2 F_3(1/2+\nu, \alpha/2+\nu;  1-\rho+\alpha/2+\nu, \nu+1, 2\nu+1,  c ^2 z^2).$\\
where the hypergeometric function
\begin{equation}
{}_2 F_3\left(a, b; c_1, c_2, c_3, z\right)=\sum_{n=0}^{\infty}\frac{(a)_n(b)_n}{(c_1)_n (c_2)_n(c_3)_n n!} z^n. \end{equation}

\newcommand{\Addresses}{{
  \bigskip
  \footnotesize

  M.V. Ould Moustapha, \textsc{Department of Mathematic,
 College of Arts and Sciences-Gurayat,
 Jouf University-Kingdom of Saudi Arabia.}\par\nopagebreak
\textsc{ Faculte des Sciences et Techniques
Universit\'e de  Nouakchott Al-asriya.
Nouakchott-Mauritanie.}\\
  \textit{E-mail address}, M. V.~Ould Moustapha: \texttt{mohamedvall.ouldmoustapha230@gmail.com}

}}

\Addresses

\end{document}